\documentclass[11pt]{article}
\usepackage{amsmath,amsthm,amssymb}
\usepackage[frenchb,english]{babel}
\usepackage{algorithmic}
\usepackage{url}
\usepackage[utf8]{inputenc} 

\newtheorem{theorem}{Theorem}[section]

\newtheorem{lemma}[theorem]{Lemma}

\newtheorem{exe}[theorem]{Example}
\newtheorem{rem}[theorem]{Remark}
\newtheorem{defi}[theorem]{Definition}
\newtheorem{algo}[theorem]{Algorithm}
\newtheorem{Not}[theorem]{Notation}
\numberwithin{equation}{section}

\newcommand{\R}{\mathbb{R}}

\newcommand{\Z}{\mathbb{Z}}

\newcommand{\N}{\mathbb{N}}

\newcommand{\pr}{^{\prime}}

\DeclareMathOperator{\bit}{bit}
\DeclareMathOperator{\Syl}{Syl}
\DeclareMathOperator{\Res}{Res}
\DeclareMathOperator{\Der}{Der}
\DeclareMathOperator{\Sign}{Sign}

\begin{document}

\title{A bound on the minimum of a real positive polynomial over the standard simplex}
\author{Saugata Basu, Richard Leroy, Marie-Françoise Roy}
\date{\today}
\maketitle

\begin{abstract}
We consider the problem of bounding away from $0$ the minimum value $m$ taken by a polynomial $P \in \Z \left[X_{1},\dots,X_{k}\right]$ over the standard simplex $\Delta \subset \R^{k}$, assuming that $m>0$. Recent algorithmic developments in real algebraic geometry enable us to obtain a positive lower bound on $m$ in terms of the dimension $k$, the degree $d$ and the bitsize $\tau$ of the coefficients of $P$. The bound is explicit, and obtained without any extra assumption on $P$, in contrast with previous results reported in the literature.
\end{abstract}

\section{Introduction}
\subsection{Problem statement}
Let $P \in \Z \left[X_{1},\dots,X_{k}\right]$ be a multivariate polynomial of degree $d$ taking only positive values on the $k-$dimensional simplex
\[\Delta = \Big\{x \in \R_{\geq 0}^{k}  \Big\vert \sum\limits_{i=1}^{k}x_{i} \leq 1\Big\}.\]
Let $\tau$ be an upper bound on the bitsize of the coefficients of $P$. Writing 
\[m = \mathop {\min }\limits_{\Delta} P >0,\]
we consider the problem of finding an explicit bound $m_{k,d,\tau}$ depending only on $k$, $d$ and $\tau$ such that $0<m_{k,d,\tau} < m$.
\subsection{Previous work}
Several authors have worked on this subject. There are two main approaches: Canny's gap theorem can be used, under non-degeneracy conditions (\cite{C}); in \cite{LS}, the authors use the Lojasiewicz inequality, leading to a bound in the general case, but involving a universal constant. The method presented here gives an explicit bound, with no extra assumption on $P$.
\subsection{Univariate case}
We begin with the univariate case, which contains 
some
basic ideas of the proof in the general case. This situation has already been studied in \cite{BCR}. We present here a simpler proof, leading to a slightly better bound. 
\\ \ 
\\
Consider a univariate polynomial of degree $d$
\[P = \sum\limits_{i=0}^{d}a_{i}T^{i} \in \Z \left[T\right],\]
taking only positive values on the interval $\left[0,1\right]$. Let $\tau$ be a bound on the bitsize of its coefficients.
\\ The minimum $m$ of $P$ on $\left[0,1\right]$ occurs either at $0$ or $1$, or at a point $x^{*}$ lying in the interior $] 0,1[$. The first case is trivial, as $P(0),P(1) \in \Z$, so that $m$ is clearly greater than $1$. In the second case, $P \left(x^{*}\right) = 0$, so that $m$ is a root of the resultant $R(Z) = \Res_{T}\left( P(T) -Z, P'(T) \right) \in \Z\left[Z\right]$. The resultant $R(Z)$ is the determinant of the matrix $\Syl(Z)$, where $\Syl(Z)$ is the following Sylvester matrix:
\[
{{\begin{pmatrix}
 {a_d } &  \cdots  &  \cdots  &  \cdots  & {a_1 } & {a_0  - Z} & 0 &  \cdots  & 0  \\
 0 &  \ddots  & {} & {} & {} & {} &  \ddots  &  \ddots  &  \vdots   \\
 \vdots  &  \ddots  &  \ddots  & {} & {} & {} & {} &  \ddots  & 0  \\
 0 &  \cdots  & 0 & {a_d } &  \cdots  &  \cdots  &  \cdots  & {a_1 } & {a_0  - Z}  \\
 {(d-1)a_{d-1} } &  \cdots  &  \cdots  &  \cdots  & {a_1 } & 0 &  \cdots  &  \cdots  & 0  \\
 0 &  \ddots  & {} & {} & {} &  \ddots  &  \ddots  & {} &  \vdots   \\
 \vdots  &  \ddots  &  \ddots  & {} & {} & {} &  \ddots  &  \ddots  &  \vdots   \\
 \vdots  & {} &  \ddots  &  \ddots  & {} & {} & {} &  \ddots  & 0  \\
 0 &  \cdots  &  \cdots  & 0 & {(d-1)a_{d-1} } &  \cdots  &  \cdots  &  \cdots  & {a_1 } 
\end{pmatrix}}\begin{array}{*{20}c}
{\left. \begin{array}{l}
  \vspace{1.5 cm}
 \end{array} \! \! \! \! \! \! \!  \right\}\  \ d-1}  \\
   {\left. \begin{array}{l}
  \vspace{2.1 cm}
 \end{array}\! \! \! \! \! \! \!  \! \! \! \! \! \! \! \! \! \! \! \! \!\! \right\}d}  \\
\end{array}}
\]
$R(Z) = \sum\limits_{i=0}^{d-1} r_i Z^i$ is thus a polynomial of degree $d-1$ in $Z$, whose coefficients are controlled in the following fashion:
\begin{lemma}
\label{r_uni}
For all $i \in \{0,\dots,d-1\}$, we have
\[\left|r_i\right| < 3^{-d/2}\left[2^{\tau}\sqrt{(d+1)^{3}}\right]^{d}{d-1 \choose i} \left[2^{\tau}\sqrt{d+1}-1\right]^{d-1-i}.\]
\end{lemma}
\begin{proof}
Let $\left(A_{1},\dots,A_{d-1},B_{1},\dots,B_{d}\right)$ denote the rows of the classical Sylvester matrix $\Syl(0)$. Then
\[R(Z) = \det\left(A_{1}+Ze_{d+1},\dots,A_{d-1}+Ze_{2d-1},B_{1},\dots,B_{d}\right),\]
where $\left(e_{1},\dots,e_{2d-1}\right)$ is the canonical basis of $\R^{2d-1}$. Using the multilinearity of the determinant, we can write $R(Z) = \sum\limits_{i=0}^{d-1} r_i Z^i$, where, for all $i \in \{0,\dots,d-1\}$, $r_{i}$ is a sum of ${d-1 \choose i}$ determinants of matrices built with:
\begin{itemize}
\item[-] $i$ rows among the $e_{j}$'s
\item[-] $d-1-i$ rows among the $A_{j}$'s
\item[-] the $d$ rows $B_{1},\dots,B_{d}$.
\end{itemize}
Hadamard's bound (see \cite{BPR}) implies that, for all $i$:
\begin{align*}
\left|r_i\right| & \leq {d-1 \choose i} \sqrt{\left[(d+1)\left(2^{2\tau}-1\right)\right]^{d-1-i}} \sqrt{\left[\dfrac{d(d+1)(2d+1)}{6}\left(2^{2\tau}-1\right)\right]^{d-1-i}}\\
& <  {d-1 \choose i} \left[2^{\tau}\sqrt{d+1}-1\right]^{d-1-i} \left[2^{\tau}\sqrt{\dfrac{(d+1)^{3}}{3}}\right]^{d}\\
&\leq 3^{-d/2}\left[2^{\tau}\sqrt{(d+1)^{3}}\right]^{d}{d-1 \choose i} \left[2^{\tau}\sqrt{d+1}-1\right]^{d-1-i},
\end{align*}
as claimed.
\end{proof}
Since the minimum $m$ is a root of $R(Z)$, Cauchy's bound finally implies the following theorem
\begin{theorem}
Let $P \in \Z \left[T\right]$ be a univariate polynomial of degree $d$ taking only positive values on the interval $\left[0,1\right]$. Let $\tau$ be an upper bound on the bitsize of the coefficients of $P$. Let $m$ denote the minimum of $P$ over $\left[0,1\right]$. Then
\[m> \dfrac{3^{d/2}}{2^{(2d-1)\tau} (d+1)^{2d-1/2}}.\]
\end{theorem}
\begin{proof}
If $m$ is attained at $0$ or $1$, then the result is obvious. If not, $m$ is a root of the resultant $R(Z)$. Since $R$ has at least one non-zero root ($R(m) = 0$), Cauchy's bound (see \cite{BPR}) implies
\begin{align*}
\dfrac{1}{m} & \leq {\sum\limits_{i=0}^{d-1} \left|r_i\right|}\\
& <{\sum\limits_{i=0}^{d-1}3^{-d/2}\left[2^{\tau}\sqrt{(d+1)^{3}}\right]^{d}{d-1 \choose i} \left[2^{\tau}\sqrt{d+1}-1\right]^{d-1-i}}\\
& \leq 3^{-d/2} \left[2^{\tau}\sqrt{(d+1)^{3}}\right]^{d} \left[2^{\tau}\sqrt{d+1}\right]^{d-1},
\end{align*}
from which the result follows easily.
\end{proof}
\begin{rem}
Our bound is slightly better than a recent one presented in \cite{BCR}, which was already almost sharp. Indeed, following \cite{BCR}, consider the polynomial $P_{k} = X^{d}+\left( 2^{k}X-1\right)^{2}$. Here, $\tau = 2k$ and the minimum $m_{k}$ of $P_{k}$ satifies
\[m_{k}\leq P_{k}\left( 2^{-k}\right) = 2^{-d\tau/2},\]
and thus decreases exponentially with $d$ and $\tau$.
\end{rem}
\section[Multivariate case]{Bound on the minimum of multivariate positive polynomial}
We now consider the multivariate case.
\subsection{Notation and problem statement}
The following notation will be useful:
\begin{Not}
We write $\bit(n)$ for the bitsize of an integer $n\in \N$.
\end{Not}
Let $P \in \Z \left[X_1,\dots,X_k\right]$ be a polynomial of degree $d$, and $\tau$ a bound on the bitsize of its coefficients. Moreover, assume that 
\[m =\mathop {\min }\limits_\Delta  P > 0.\]
In order to find an explicit lower bound $0< m_{k,d,\tau}<m$, we generalize the proof of the univariate case. We first show that, at the cost of slightly increasing the bitsize of the coefficients,  we can assume that the minimum is attained in the interior of the simplex. Obviously, there exists a face $\sigma$ of $\Delta$, of dimension $0 \leq s \leq k$, such that the minimum $m$ is attained at a point of the interior of $\sigma$ (with its induced topology). In the following we consider such a face $\sigma$, of minimal dimension $s$.
\begin{rem}
If $\sigma$ is a vertex of $\Delta$, then obviously $m \geq 1$. We now assume that $s \geq 1$.
\end{rem}
Denote by
\[
\left\{ \begin{array}{l}
 V_0 = 0 \\ 
 V_i = e_i \quad (1\leq i \leq k)  
 \end{array} \right.
\]
the vertices of $\Delta$, and
\[
\left\{ \begin{array}{l}
 \lambda_0 = 1 - \sum X_i \\ 
 \lambda_i = X_i \quad (1\leq i \leq k)
 \end{array} \right.
\]
the associated barycentric coordinates.
\\
\\There exists a subset $I = \{i_0,\dots,i_{s} \}$ of $\{0,\dots,k\}$ such that the vertices of $\sigma$ are the vertices $\left(V_i\right)_{i \in I}$. Let $J=\{0,\dots,k\} \backslash I$. The face $\sigma$ is characterized by:
\[\sigma = \{ x \in \Delta \left|\forall j \in J, \  \lambda_j(x) = 0\right. \}.\]
We make the following substitutions in $P$: 
\\
\\ $\bullet$ If $j \in J$ and $j > 0$, replace the variable $X_j$ by $0$
\\ $\bullet$ If $j \in J$ and $j = 0$, replace the variable $X_{i_0}$ by $1 - \sum\limits_{\ell = 1}^{s} {X_{i_\ell}}$
\\
\\ We then obtain a polynomial $P_\sigma \in \Z \left[X_{i_1},\dots,X_{i_{s}}\right]$ satisfying :
\[
\mathop {\min }\limits_{\Delta} P = \mathop {\min }\limits_{\mathop \sigma \limits^ \circ} P_\sigma.
\]
Renaming the variables $X_{i_\ell}$ into $Y_\ell$, we obtain that $P_\sigma \in \Z \left[Y_1,\dots,Y_s\right]$ is a polynomial verifying:
\begin{lemma}
The degree of $P_\sigma$ is bounded by $d$.
\\Moreover, the bitsize of its coefficients is bounded by $\tau_\sigma$, where
\begin{equation}
\label{tausigma}
\tau_\sigma = \tau + 1 + d \bit (k).
\end{equation}
\end{lemma}
\begin{proof}
The degree of $P_\sigma$ is clearly at most $d$. We now show the result concerning the bitsizes of the coefficients.

The result is clear if $s= k$.

Assume that $1 \leq s \leq k-1$.
\\Since replacing $X_j$ by $0$ does not change the bound on the bitsize of the coefficients, only the replacement of $Y_{0}$ by $1 - \sum\limits_{i = 1 }^{s} {Y_i}$ has to be taken into account.
\\
\\If \[P = \sum\limits_{\substack{\alpha \in \N^k \\ |\alpha| \leq d}}a_\alpha X^\alpha,\]
then
\[P_\sigma = \sum\limits_{\substack{\gamma \in \N^s \\ |\gamma| \leq d}}b_\gamma Y^\gamma,\]
where
\[b_\gamma = \sum\limits_{\substack{\beta \in I_\gamma}} \pm { |\beta| \choose \beta} a_{\left(\gamma_1-\beta_1,\dots,\gamma_k-\beta_k\right)},\]
and \[I_\gamma = \{ \beta \in \N^{s+1} \Big\vert |\beta| \leq d  \text{ and } \forall i \in \{1,\dots,s\}, \  \beta_i \leq \gamma_i\}.\]
Hence, we have
\begin{eqnarray*}
|b_\gamma| &\leq& 2^\tau \sum\limits_{\substack{\beta \in I_\gamma}} { |\beta| \choose \beta} \\
& \leq &  2^\tau \sum\limits_{\substack{\beta \in \N^{s+1}\\|\beta| \leq d}}{ |\beta| \choose \beta} \\
& \leq & 2^\tau \sum\limits_{p=0}^{d}\sum\limits_{\substack{\beta \in \N^{s+1}\\|\beta| =p}} { |\beta| \choose \beta} \\
& \leq & 2^\tau \sum\limits_{p=0}^{d} (s+1)^p \\
& \leq & 2^\tau \dfrac{(s+1)^{d+1}}{s} \\
& \leq & 2^\tau\times 2(s+1)^d \\
& \leq & 2^{\tau+1}k^d,
\end{eqnarray*}
and the conclusion follows.
\end{proof}
Since
\[
m=\mathop {\min }\limits_{\Delta} P = \mathop {\min }\limits_{\mathop \sigma \limits^ \circ} P_\sigma,
\]
$P_\sigma \in \Z \left[Y_1,\dots,Y_s\right]$ achieves its minimum in the interior of $\sigma$. Consequently, $m$ is attained at a critical point of $P_\sigma$, $i.e.$ a point $x\in \R^s$ such that the gradient of $P_\sigma$ is zero at $x$. We are thus interested in computing the values of $P$ at the zeros of its gradient. We aim at giving a univariate reformulation of this problem, enabling us to use resultant methods. The following section introduces the necessary material.
\subsection{Rational univariate representation}
We first introduce the notion of Thom encoding:
\begin{defi}
Let $P \in \R[X]$ be a real univariate polynomial, $x\in\R$ a real number and $\sigma \in \{0,1,-1\}^{\Der(P)}$ a sign condition on the set $\Der(P) = \left\{P,P\pr,\dots,P^{(\deg P)}\right\}$ of the derivatives of $P$.
\\
\\ The sign condition $\sigma$ is a Thom encoding of $x$ if $\sigma(P) = 0$ and 
\[\forall i,\quad \Sign \left(P^{(i)}(x)\right) = \sigma \left(P^{(i)}\right).\]
\end{defi}
We can now define a rational univariate representation as follows:
\begin{defi}
An $s-$rational univariate representation $u$ is an $(s+3)-$tuple of the form
\[u= \left(F(T),g_0(T),\dots,g_s(T),\pi \right)\]
such that:
\begin{enumerate}
\item $F,g_0,\dots,g_s \in \R \left[T\right]$,
\item $F$ and $g_0$ are coprime,
\item $\pi$ is a Thom encoding of a root $t_\pi \in \R$ of $F$.
\end{enumerate}
\end{defi}
\begin{rem}
If $t \in \R$ is a root of $F$, then $g_0(t) \neq 0$.
\end{rem}
We now define the point associated to the rational univariate representation:
\begin{defi}
The point associated to $u$ is defined by
\[x_u = \left(\dfrac{g_1(t_\pi)}{g_0(t_\pi)},\dots,\dfrac{g_s(t_\pi)}{g_0(t_\pi)}\right).\]
\end{defi}
Hence, a rational univariate representation gives rise to a point whose coordinates are rational fractions evaluated at a root of $F$.
\\
\\
Let $Q \in \R^s$ be a nonnegative polynomial over $\R^s$, and
\[\mathcal{Z}(Q) =\big\{ x \in \R^s \big\vert Q(x) = 0 \big\}\]
be the set of real zeros of $Q$. We are interested in finding a point in each bounded connected component of $\mathcal{Z}(Q)$. This can be done by applying Algorithm 12.15 of \cite{BPR}, which we recall here for convenience.
\begin{algo}[Bounded Algebraic Sampling]
\begin{algorithmic}
\label{bounded}
\STATE
\STATE
\REQUIRE A polynomial $Q \in \Z \left[X_1,\dots X_s\right]$, of degree bounded by $d_Q$, nonnegative over $\R^s$.
\ENSURE A set $\mathcal{U}$ of rational univariate representations of the form
\[\left(F(T),g_0(T),\dots,g_s(T),\pi \right),\]
where the polynomials $F,g_0,\dots,g_s$ have integer coefficients, and such that the associated points meet every bounded connected component of $\mathcal{Z}(Q)$.
\STATE 
\end{algorithmic}
\end{algo}
We indicate the main ideas behind the algorithm, referring the reader to \cite{BPR} for details:
\begin{itemize}
\item replace $Q$ by a deformation ${\rm Def}(Q,d_Q,\zeta)$ of degree bounded by $d_Q+2$, where $\zeta$ is an infinitesimal,
\item consider the critical points ${\rm Cr}({\rm Def}(Q,d_Q,\zeta))$ of ${\rm Def}(Q,d_Q,\zeta)$ in the $X_1$-direction,
\item due to the properties of ${\rm Def}(Q,d_Q,\zeta)$,
\begin{itemize}
\item ${\rm Cr}({\rm Def}(Q,d_Q,\zeta))$ has a finite number of points,
\item the quotient ring defined by the equations of ${\rm Cr}({\rm Def}(Q,d_Q,\zeta))$ is a  vector space of dimension at most 
\[(d_Q+2)(d_Q+1)^{k-1},\]
\item its multiplication table can be easily computed,
\end{itemize}
\item find rational univariate representations of the points of ${\rm Cr}({\rm Def}(Q,d_Q,\zeta))$,
\item take their limits with respect to $\zeta$, which define a finite set of points intersecting all the bounded connected components of $\mathcal{Z}(Q)$.
\end{itemize}
The complexity analysis in \cite{BPR} shows that, if $d_Q$ is a bound on the degree of $Q$ and $\tau_Q$ a bound on the bitsize of its coefficients, then:
\begin{enumerate}
\item The degrees of the polynomials $F,g_0,\dots,g_k$ are bounded by \[(d_Q+2)(d_Q+1)^{k-1}\]
\item The bitsize of their coefficients is bounded by
\[(d_Q+2)(d_Q+1)^{k-1}(kd_Q+2)  \left(\tau \pr +2 \bit(kd_{Q}+3) + 3\mu + \bit(k)\right),\]
where
\begin{eqnarray*}
\tau {\pr } &=& \sup \left[\tau_Q,\bit(2k)\right] + 2 \bit \left[k (d_Q+2)\right] +1 \\
\mu &=& \bit \left[(d_Q+2)(d_Q+1)^{k-1}\right].
\end{eqnarray*}
\end{enumerate}
\subsection{The bound}
Recall that $P_\sigma \in \Z \left[Y_1,\dots,Y_s\right]$ achieves its minimum in the interior of $\sigma$. Consequently, this minimum is attained at a critical point of $P_\sigma$, $i.e.$ a point $x\in \R^s$ at which the gradient of $P_\sigma$ is zero. Consider the set of critical points
\[\mathcal{Z} = \left\{x \in \R^s \Big\vert {\dfrac{\partial P_\sigma}{\partial Y_1}} (x)= \dots = {\dfrac{\partial P_\sigma}{\partial Y_s}} (x) = 0\right\}.\]
Note that $P_{\sigma}$ is constant on each connected component of $\mathcal{Z}$. So, we aim at computing a set $\mathcal{U}$ of rational univariate representations $u$ whose associated points $x_{u}$ meet every connected component of $\mathcal{Z}$, together with the values  $P_{\sigma}(x_{u})$.
\begin{rem}
When $\mathcal{Z}$ has a finite number of points, then Gr\"obner basis techniques can be used to obtain rational univariate representations of these points (see \cite{R}). The method we present hereafter, based on Algorithm \ref{bounded}, computes a point in every connected component of $\mathcal{Z}$ even if $\mathcal{Z}$ is infinite. Moreover Algorithm \ref{bounded} makes it possible to control the degree and the bitsize of the coefficients of the output, in contrast with Gr\"obner basis methods.
\end{rem}
It is easy to see that if $C$ is a connected component of $\mathcal{Z}$ containing a minimizer of $P_\sigma$ in $\sigma$, then $C \subset {\mathop \sigma \limits^ \circ}$ by minimality of the dimension $s$ of $\sigma$. In particular, $C$ is bounded. Algorithm \ref{bounded} then gives a set of rational univariate representations of the form 
\[u = \left(F(T), g_0(T),g_1(T),\dots,g_s(T),\pi\right),\]
whose associated points meet every bounded connected component of $\mathcal{Z}$. In particular, they meet every connected component of $\mathcal{Z}$ containing a minimizer of $P_\sigma$ in $\sigma$.
\begin{lemma}
\label{rur}
The degree of the polynomials $F,g_0,\dots,g_s$ is bounded by $d_u$, where
\[d_u = 2d(2d-1)^{k-1}.\]
Moreover,  the bitsize of their coefficients is bounded by
\[\tau_u = d_{u}(2kd-2k+2) \left[\tau \pr + 2\bit(2kd-2k+3)+3\bit(d_{u}) + \bit(k)\right],\]
where
\begin{eqnarray*}
\tau \pr &=& 2\tau +(2d+2)\bit(k) + (k+3)\bit(d) + 5.
\end{eqnarray*}
\end{lemma}
\begin{proof}
Let $Q$ denote the polynomial \[Q = \sum\limits_{i=1}^{s} {\left(\dfrac{\partial P_\sigma}{\partial Y_i}\right)^2}.\]
Clearly, its degree is bounded by $d_Q = 2d-2$. Moreover, we can bound the bitsize of its coefficients as follows.
\\ If \[P_\sigma = \sum\limits_{\substack{\gamma \in \N^s \\ |\gamma| \leq d}}b_\gamma Y^\gamma,\]
then
\[{\left(\dfrac{\partial P_\sigma}{\partial Y_i}\right)^2} = \sum\limits_{\substack{\gamma \in \N^s \\ |\gamma| \leq d}} c_\gamma Y^{\gamma-2e_i},\]
where
\[c_\gamma = \sum\limits_{\substack{\alpha \in \N^s \\ \alpha \leq \gamma}}\alpha_i (\gamma_i-\alpha_i)a_\alpha a_{\gamma-\alpha}.\]
Write $ Q =\sum\limits_{\substack{\delta \in \N^s \\ |\delta| \leq d-2}}d_\delta Y^\delta$. Since $Q = \sum\limits_{i=1}^{s} {\left(\dfrac{\partial P_\sigma}{\partial Y_i}\right)^2},$ its coefficients are bounded as follows:
\begin{eqnarray*}
\left|d_\delta \right| & \leq & s \sum\limits_{\substack{\alpha \in \N^s \\ \alpha \leq \gamma}}\alpha_i(\gamma_i-\alpha_i)a_\alpha a_{\gamma-\alpha} \\
& \leq & sd^22^{2\tau_\sigma}d^k \\
& \leq & k2^{2\tau_\sigma}d^{k+2}.
\end{eqnarray*}
Hence, the bitsize of the coefficients of $Q$ is bounded by $\tau_Q$, where
\[\tau_Q = 2\tau_\sigma + (k+2)\bit(d) + \bit(k) = 2\tau +(2d+1)\bit(k) + (k+2)\bit(d) + 2,\]
where the last equality follows from equation \eqref{tausigma}.
The result now follows from the complexity analysis of Algorithm \ref{bounded}.
\end{proof}
Let $P_u (T) =  g_0(T)^d P_\sigma \left(\dfrac{g_1(T)}{g_0(T)},\dots,\dfrac{g_s(T)}{g_0(T)}\right)$. We have:
\begin{lemma}
\label{Pu}
The degree of $P_u$ is bounded by $d_{P,u} = d_u d.$
\\The bitsize of its coefficients is bounded by $\tau_{P,u}$, where
\[\tau_{P,u} =d \left[\tau_u + \bit(d_u+1)\right] + \tau + d\bit(k) + d + k + 1.\]
\end{lemma}
\begin{proof}
The result about the degree is clear from the previous lemma.
\\The bound on the bitsize of the coefficients is obtained by substitution, using Proposition $8.11$ of \cite{BPR}.
\end{proof}
The minimum $m$ of $P_\sigma$ over $\sigma$ is attained at a point $x\in \sigma$ contained in a connected component of $\mathcal{Z}$ included in the ball $B(0,1)$. Since $P_\sigma$ is constant on such a component, $m$ is also attained at some point $x_u$ associated to an already computed rational univariate representation $ u = \left(F(T), g_0(T),g_1(T),\dots,g_s(T),\pi\right)$.

Since $t_\pi$ is a root of $F$, the minimum $m = P_\sigma \left(x_u\right)$ is a root of the resultant 
\[R(Z) = \Res_T \left(P_u(T)-g_{0}(T)^{d}Z,F(T)\right).\]
\begin{exe}
We consider here the following easy example (Berg polynomial, see Example $37$ in \cite{Sc}):
\[B:=x^2y^{2}(x^{2}+y^{2}-1)+1.\]
It is easy to show that $B$ is positive on $\Delta$. We now compute its minimum.
\begin{itemize}
\item On the three vertices of $\Delta$, we have $B =1$.
\item On the faces $\{x=0\}$ and $\{y=0\}$, we have $B=1$.
\item Consider the face $\{x+y=1\}$. Replacing $x$ by $1-y$ leads to consider the (univariate) polynomial
\[ B_{\{x+y=1\}} = 2y^{6}-6y^{5}+6y^{4}-2y^{3}+1.
\]
Since $B'_{\{x+y=1\}} = 6y^{2}(y-1)^{2}(2y-1)$, the minimum of $B_{\{x+y=1\}}$ is $31/32$, attained at $y=1/2$.
\item We now compute the values of $B$ at its critical points contained in the interior of $\Delta$. It is easy to show that those points $(x,y)$ satisfy
\begin{align*}
2x^{2}+y^{2}&=1\\
x^{2}+2y^{2}&=1.
\end{align*}
Note that this easily implies that there is only a finite number of such critical points. A rational univariate representation of this set can then be computed (using for example Salsa software, see \cite{Sa}):
\begin{align*}
F&= (3T^{2}-1)(T^{2}-3)\\
g_{0} &= T(3T^{2}-5)\\
g_{1} &= T^{2}+1\\
g_{2} &= 2(T^{2}-1).
\end{align*}
The resultant $R(Z)$ is equal to
\begin{align*}
R(Z) &= \Res_{T}\left( g_{0}(T)^{6}B\left( \dfrac{g_{1}(T)}{g_{0}(T)},\dfrac{g_{2}(T)}{g_{0}(T)}\right) - Z g_{0}(T)^{6},F(T)\right)\\
&= 2^{48}3^{6}(27Z-26)^{4}.
\end{align*}
The only root $26/27< \min(1,31/32)$ is thus the minimum of $B$ over $\Delta$, corresponding to the root $\sqrt{3}$ of $F$, and giving the minimizer
\[\left( \dfrac{g_{1}(\sqrt{3})}{g_{0}(\sqrt{3})},\dfrac{g_{2}(\sqrt{3})}{g_{0}(\sqrt{3})}\right) =\left( \dfrac{1}{\sqrt{3}},\dfrac{1}{\sqrt{3}}\right).\]
\end{itemize}
\end{exe}
In order to obtain a lower bound on the minimum depending only on $k,d$ and $\tau$, one needs to bound the roots of $R(Z)$. This can be done by controlling the size of the coefficients of $R(Z)$ and then using Cauchy's bound. Write
\[F(T) = \sum\limits_{i=0}^{d_u} f_i T^i,\]
\[P_u (T) =\sum\limits_{i=0}^{d_{P,u}}  a_i T^i\]
and
\[g_{0}(T)^d = \sum\limits_{i=0}^{d_{u}d}  b_i T^i = \sum\limits_{i=0}^{d_{P,u}}  b_i T^i.\]
\begin{lemma}
The polynomial $g_{0}(T)^d$ has degree bounded by $d_{u}d=d_{P,u}$, and the bitsize of its coefficients is bounded by $d(\tau_{u}+\bit(d_{u}+1))\leq \tau_{P,u}$.
\end{lemma}
\begin{proof}
The degree of $g_{0}(T)^d$ is clearly less than $d_{u}d=d_{P,u}$. We now show the bound on the bitsize of its coefficients.
Recall that the degree of $g_{0}$ is bounded by $d_{u}$and that the bitsize of its coefficients is less than$\tau_{u}$. When multiplying a univariate polynomial $f$ by $g_{0}$, the increase in the bitsize of the coefficients is at most $\tau_{u} + \bit(d_{u}+1)$. Indeed, the coefficients of $fg_{0}$ are sums of at most $(d_{u}+1)$ products of a coefficient of $f$ by a coefficient of $g_{0}$. The conclusion follows easily.
\end{proof}
The resultant $R(Z)$ is the determinant of the matrix $\Syl(Z)$, where $\Syl(Z)$ is the following Sylvester matrix:
\[
{\begin{pmatrix}
   {a_{d_{P,u}}-b_{d_{P,u}}Z } &  \cdots  &  \cdots  &  \cdots  & {\cdots} & {a_0  - b_{0}Z} & 0 &  \cdots  & 0  \\
   0 &  \ddots  & {} & {} & {} & {} &  \ddots  &  \ddots  &  \vdots   \\
    \vdots  &  \ddots  &  \ddots  & {} & {} & {} & {} &  \ddots  & 0  \\
   0 &  \cdots  & 0 & {a_{d_{P,u}}-b_{d_{P,u}}Z  } &  \cdots  &  \cdots  &  \cdots  & {\cdots} & {a_0  -b_{0}Z }  \\
   {f_{d_{u}} } &  \cdots  &  \cdots  &  \cdots  & {f_0 } & 0 &  \cdots  &  \cdots  & 0  \\
   0 &  \ddots  & {} & {} & {} &  \ddots  &  \ddots  & {} &  \vdots   \\
    \vdots  &  \ddots  &  \ddots  & {} & {} & {} &  \ddots  &  \ddots  &  \vdots   \\
    \vdots  & {} &  \ddots  &  \ddots  & {} & {} & {} &  \ddots  & 0  \\
   0 &  \cdots  &  \cdots  & 0 & {f_{d_{u}} } &  \cdots  &  \cdots  &  \cdots  & {f_0 } 
\end{pmatrix}} \begin{array}{*{20}c}
   {\left. \begin{array}{l}
\vspace{1.5 cm}
 \end{array}  \!\!\!\!\!\!\!\!\!\!\!\!\!\!\!\right\}d_u}  \\
   {\left. \begin{array}{l}
  \vspace{2.1 cm}
 \end{array} \!\!\!\!\!\!\!\!\!\!\! \right\}d_{P,u}}  \\
\end{array}
\]
$R(Z) = \sum\limits_{i=0}^{d_u} r_i Z^i$ is a polynomial of degree $\leq d_u$ in $Z$, whose coefficients are controlled in the following fashion:
\begin{lemma}
For all $i \in \{0,\dots,d_u\}$,
\[\left|r_i\right| < { d_u \choose i} \left[2^{{\tau_{P,u}}}\sqrt{d_{P,u}+1}\right]^{d_u} \left[2^{{\tau_u}}\sqrt{d_u+1}\right]^{d_{P,u}}.\]
\end{lemma}
\begin{proof}
We proceed as in the proof of lemma \ref{r_uni}.
\\ Let (with obvious notation)
\[\left(A_{1}+ZB_{1},\dots,A_{d_{u}}+ZB_{d_{u}},C_{1},\dots,C_{d_{P,u}}\right)\]
be the rows of the Sylvester matrix $\Syl(Z)$. Using the multilinearity of the determinant, we can write $R(Z) = \sum\limits_{i=1}^{d_{u}}r_{i}Z^{i}$, where, for all $i \in \{0,\dots,d_{u}\}$, $r_{i}$ is a sum of ${d_{u} \choose i}$ determinants of matrices built with:
\begin{itemize}
\item[-] $i$ rows among the $B_{j}$'s
\item[-] $d_{u}-i$ rows among the $A_{j}$'s
\item[-] the $d_{P,u}$ rows $C_{1},\dots,C_{d_{P,u}}$.
\end{itemize}
Hadamard's bound (see \cite{BPR}) implies that, for all $i$:
\begin{align*}
\left|r_i\right| & \leq{d_{u} \choose i} \sqrt{\left[(d_{P,u}+1)\left(2^{2\tau_{P,u}}-1\right)\right]^{d_{u}}}\sqrt{\left[(d_{u}+1)\left(2^{2\tau_{u}}-1\right)\right]^{d_{P,u}}}\\
& <  { d_u \choose i} \left[2^{{\tau_{P,u}}}\sqrt{d_{P,u}+1}\right]^{d_u} \left[2^{{\tau_u}}\sqrt{d_u+1}\right]^{d_{P,u}},
\end{align*}
as claimed.
\end{proof}
Since the minimum $m$ is a root of $R(Z)$, Cauchy's bound finally gives the estimate we were looking for.
\\

Let $\mathcal{U}$ be a set of rational univariate representations whose associated points meet every bounded connected component of
\[\mathcal{Z}_\sigma =  \left\{x \in \R^s  \Big\vert \dfrac{\partial P_\sigma}{\partial Y_1}(x) = \dots = \dfrac{\partial P_{\sigma}}{\partial Y_{s}}(x) = 0 \right\},\]
for each face $\sigma$ of $\Delta$.

Also, let $d_u$ (resp. $\tau_u$) be a bound on the degree (resp. the bitsize of the coefficients) of the polynomials occuring in the rational univariate representations of $\mathcal{U}$, and  $d_{P,u}$ (resp. $\tau_{P,u}$) be
a bound on the degree (resp. the bitsize of the coefficients) of the polynomial $P_{u}$.
Then:
\begin{theorem}
\label{th_minimum}

\[m > \dfrac{1}{\left[2^{{\tau_{P,u}+1}}\sqrt{d_{P,u}+1}\right]^{d_u}\left[2^{{\tau_u}}\sqrt{d_u+1}\right]^{d_{P,u}}}.\]
\end{theorem}
\begin{proof}
Since $R$ has at least one non-zero root ($R(m) = 0$), we can write 
\[R(Z) = \sum\limits_{i=q}^{p} r_i Z^i,\]
with $q<p\leq d_{u}$ and $r_q r_{p} \neq 0$.
\\
\\ Cauchy's bound then implies:
\begin{eqnarray*}
\dfrac{1}{m} & \leq & {\sum\limits_{i=0}^{d_u} \left|r_i\right|} \\
& < & {\sum\limits_{i=0}^{d_u} { d_u \choose i} \left[2^{{\tau_{P,u}}}\sqrt{d_{P,u}+1}\right]^{d_u} \left[2^{{\tau_u}}\sqrt{d_u+1}\right]^{d_{P,u}}}\\
& \leq & {2^{d_{u}}\left[2^{{\tau_{P,u}}}\sqrt{d_{P,u}+1}\right]^{d_u} \left[2^{{\tau_u}}\sqrt{d_u+1}\right]^{d_{P,u}}}\\
& \leq & {\left[2^{{\tau_{P,u}+1}}\sqrt{d_{P,u}+1}\right]^{d_u}\left[2^{{\tau_u}}\sqrt{d_u+1}\right]^{d_{P,u}}},
\end{eqnarray*}
as announced.
\end{proof}
Using Lemmas \ref{rur} and \ref{Pu}, Theorem \ref{th_minimum} immediately implies:
\begin{theorem}
Let $P \in \Z \left[X_{1},\dots,X_{k}\right]$ be a polynomial of degree $d$, $\tau$ a bound on the bitsize of its coefficients and $m = \mathop {\min }\limits_\Delta  P$ the minimum of $P$ over the simplex $\Delta$. Assume that $m > 0$. \\ Let 
\begin{align*}
D &= 2d(2d-1)^{k-1},\\
\rho &= D(2kd-2k+2) \left[\tau \pr + 2\bit(2kd-2k+3)+3\bit(D) + \bit(k)\right],\\
\rho \pr &= d \left[\rho + \bit(D+1)\right]  + \tau + d\bit(k)+ d + k + 1,
\end{align*}
where
\begin{eqnarray*}
\tau \pr &=&2\tau +(2d+2)\bit(k) + (k+3)\bit(d) + 5.
\end{eqnarray*}
Then:
\[m > m_{k,d,\tau},\]
where
\[m_{k,d,\tau} = \dfrac{1}{\left[2^{{\rho \pr+1}}\sqrt{dD+1}\right]^{D}\left[2^{{\rho}}\sqrt{D+1}\right]^{dD}}.\]
\end{theorem}
\begin{rem}
We now give a more compact bound, derived from the last corollary. This will enable us to compare our results to those of de Loera and Santos (\cite{LS}) and Canny (\cite{C}).
\\The following estimates
\begin{align*}
D+1&\leq 2^{k}d^{k}\\
\frac{\rho}  {2^{k+1}d^{k+1}k} &\leq  2\tau +(2d+5)\bit(k)+(4k+5)\bit(d)+6k+9\\
\frac{\rho \pr} { 2^{k+1}d^{k+2}k}&\leq  2\tau +(2d+5)\bit(k)+(4k+5)\bit(d)+6k+9
\end{align*}
lead to the bound:
\begin{equation}
\label{bound}
\dfrac{1}{m_{k,d,\tau}} \leq \left(2^{\tau}\right)^{2^{k+3}d^{k+1}k}2^{2^{k+6}d^{k+2}k^{2}}k^{2^{k+5}d^{k+2}k}d^{2^{k+5}d^{k+1}k^{2}}.
\end{equation}

In \cite{LS}, the authors give the following estimate:
\begin{equation*}
\label{ls}
 \dfrac{1}{m_{k,d,\tau}} \leq \left(2^{\tau}\right)^{B^{c(k+1)}}2^{B^{c(k+1)}}, \tag{\dag}
\end{equation*}
where $B$ denotes a bound on $\max(d+1,k+1)$ and $c$ is an (unknown) universal constant.
Note that bound \eqref{bound} implies that
\[ \dfrac{1}{m_{k,d,\tau}} \leq \left(2^{\tau}\right)^{2^{k+3}d^{k+1}k}2^{2^{k+7}d^{k+2}k^{2}},
\]
giving an explicit version of estimate \eqref{ls}.
\\

Moreover, a direct application of Canny's theorem (\cite{C}), under a nondegeneracy assumption on the following polynomial system (with unknowns $m,Y_{1},\dots,Y_{s}$)
\[\begin{cases}
\label{sys}
P_{\sigma}(Y_{1},\dots,Y_{s}) =m\\
{\dfrac{\partial P_\sigma}{\partial Y_1}}(Y_{1},\dots,Y_{s})= \dots ={\dfrac{\partial P_\sigma}{\partial Y_s}}(Y_{1},\dots,Y_{s})= 0,
\end{cases} \tag{S}\]
leads to the estimate
\begin{equation*}
\dfrac{1}{m_{k,d,\tau}} \leq \big(3d2^{\tau_{\sigma}}\big)^{(k+1)d^{k+1}}= \big(6dk^{d}2^{d}2^{\tau}\big)^{(k+1)d^{k+1}}. \label{c} \tag{\ddag}
\end{equation*}
In the general case, \eqref{bound}  implies 
\[ \dfrac{1}{m_{k,d,\tau}} \leq \big(dk^{d}2^{d}2^{\tau}\big)^{2^{k+5}k^{2}d^{k+1}},
\]
where the main difference with \eqref{c} is the presence of the exponent  $2^{k+5}$. This essentially comes from doubling the degree of the system \eqref{sys} by replacing the equations
\[{\dfrac{\partial P_\sigma}{\partial Y_1}}= \dots ={\dfrac{\partial P_\sigma}{\partial Y_s}}= 0
\]
by the following single one
\[Q = \sum\limits_{i=1}^{s} {\left(\dfrac{\partial P_\sigma}{\partial Y_i}\right)^2} = 0,\]
in order to cover the degenerate cases as well.

\end{rem}

\end{document}